\documentclass[US-letter,USenglish,authorcolumns,autoref]{lipics-v2019}

\usepackage{graphicx}
\usepackage{wrapfig}
\usepackage{amssymb}
\usepackage{amsmath}
\usepackage{bm}
\usepackage{latexsym}
\usepackage{graphicx}
\usepackage{booktabs}
\usepackage{float}
\usepackage{comment}
\usepackage[shortlabels]{enumitem} 
\usepackage[utf8]{inputenc} 
\usepackage{listings}
\usepackage[dvipsnames]{xcolor}
\usepackage[noend]{algpseudocode}
\usepackage{algorithm,algorithmicx}
\usepackage{url}
\usepackage[font=footnotesize]{caption} 
\usepackage[misc]{ifsym} 

\newcommand{\blue}[1] {\textcolor{blue}{#1}}

\definecolor{shadecolor}{rgb}{.9, .9, .9}
\usepackage{framed} 
\colorlet{exframecolor}{orange}
    {\endMakeFramed}

    \newenvironment{frshaded*}{%
    \MakeFramed {\advance\hsize-\width \FrameRestore}}%
    {\endMakeFramed}
    
    \newcounter{examplecounter}
\newenvironment{exam}{
 \begin{frshaded*}
    \refstepcounter{examplecounter}%
    \noindent
  \textbf{Example \arabic{examplecounter}}%
  \quad
}{%
\end{frshaded*}
}

{\endMakeFramed}

\newenvironment{frshaded2*}{%
    \MakeFramed {\advance\hsize-\width \FrameRestore}}%
    {\endMakeFramed}
\newenvironment{result}{
 \begin{frshaded2*}
}{%
\end{frshaded2*}
}
{\endMakeFramed}

\newenvironment{frshaded3*}{%
    \MakeFramed {\advance\hsize-\width \FrameRestore}}%
    {\endMakeFramed}

\newcommand{\Gen}{\textsc{Gen}}
\newcommand{\RevGen}{\textsc{RevGen}}
\newcommand{\LIST}[1]{\mathcal{G}_{#1}}
\newcommand{\REVLIST}[1]{\mathcal{R}_{#1}}
\newcommand{\spt}[1]{\ensuremath{t_{#1}}}

\algrenewcommand\algorithmicrequire{\textbf{Precondition:}}
\algrenewcommand\algorithmicensure{\textbf{Postcondition:}}

\definecolor{verbgray}{gray}{0.9}
\lstnewenvironment{code}{%
  \lstset{backgroundcolor=\color{verbgray},
 frame=single,
  language=C,
  framerule=0pt,
  basicstyle=\ttfamily,
  showstringspaces=false,
  commentstyle=\color{blue}\textit,
  keywordstyle=\color{black}\bf,
  numbers=none,
  keepspaces=true,
  columns=fullflexible}}{}

\nolinenumbers

\title{Pivot Gray Codes for the Spanning Trees of a Graph ft.~the Fan}
\titlerunning{A Pivot Gray Code for Spanning Trees}

\author{Ben Cameron}{The King's University, Canada}{ben.cameron@kingsu.ca}{}{}
\author{Aaron Grubb}{University of Guelph, Canada}{agrubb@uoguelph.ca}{}{}
\author{Joe Sawada}{University of Guelph, Canada}{jsawada@uoguelph.ca}{}{}
\authorrunning{B. Cameron and A. Grubb and J. Sawada} 
\authorrunning{~} 

\Copyright{}
\Copyright{~}


\begin{document}







%
%

\maketitle   
\vspace{0.0in}

\begin{abstract}

We consider the problem of listing all spanning trees of a graph $G$ such that successive trees differ by pivoting a single edge around a vertex.  Such a listing is called a ``pivot Gray code'', and it has more stringent conditions than known ``revolving-door'' Gray codes for spanning trees. Most revolving-door algorithms employ a standard edge-deletion/edge-contraction recursive approach which we demonstrate presents natural challenges when requiring the ``pivot'' property.  Our main result is the discovery of a greedy strategy to list the spanning trees of the fan graph in a pivot Gray code order. It is the first greedy algorithm for exhaustively generating spanning trees using such a minimal change operation. The resulting listing is then studied to find a recursive algorithm that produces the same listing in $O(1)$-amortized time using $O(n)$ space.  Additionally, we present $O(n)$-time algorithms for ranking and unranking the spanning trees for our listing; an improvement over the generic $O(n^3)$-time algorithm for ranking and unranking spanning trees of an arbitrary graph. Finally, we discuss how our listing can be applied to find a pivot Gray code for the wheel graph.

\end{abstract}


\vspace{0.0in}

\section{Introduction} \label{sec:intro}

Applications of efficiently listing all spanning trees of general graphs are ubiquitous in computer science and also appear in many other scientific disciplines \cite{chakraborty}. In fact, one of the earliest known works on listing all spanning trees of a graph is due to the German physicist Wilhelm Feussner in 1902 who was motivated by an application to electrical networks \cite{1902}. In the 120 years since Feussner's work, many new algorithms have been developed, such as those in the following citations \cite{berger,Char,cummins,gabow,hakimi,holzmann,kamae,kapoor,kishi,matsui,mayeda,minty,Shioura1995,uno,smith1997generating,winter}. 

For any application, it is desirable for spanning tree listing algorithms to have the asymptotically best possible running time, that is, $O(1)$-amortized running time. The algorithms due to Kapoor and Ramesh \cite{kapoor}, Matsui \cite{matsui}, Smith \cite{smith1997generating}, Shioura and Tamura \cite{Shioura1995} and Shioura et al. \cite{uno}  all run in $O(1)$-amortized time. Another desirable property of such listings is to have the \emph{revolving-door} property, where successive spanning trees differ by the addition of one edge and the removal of another. Such listings where successive objects in a listing differ by a constant number of simple operations are more 
generally known as \textit{Gray codes}. 
The algorithms due to Smith \cite{smith1997generating}, Kamae \cite{kamae}, Kishi and Kajitani \cite{kishi}, Holzmann and Harary \cite{holzmann} and Cummins \cite{cummins} all produce Gray code listings of spanning trees for an arbitrary graph.  Of all of these algorithms, Smith's is the only one that produces a Gray code listing in $O(1)$-amortized time.  

\begin{exam}
Consider the fan graph on five vertices illustrated in Figure~\ref{fig:F5}, where the seven edges are labeled  \footnotesize
$$ e_1=v_1v_2, \ e_2= v_1v_{\infty}, \  e_3= v_2v_3, \ e_4= v_2v_{\infty}, \  e_5= v_3v_4, \ e_6=v_3v_{\infty}, \ e_7=v_4v_{\infty}.$$
\normalsize
The following is a revolving-door Gray code for the 21 spanning trees of this graph. The initial spanning tree has edges $\{e_1,e_2,e_5,e_6\}$ and each step of the listing below provides the edge that is removed from the current tree followed by the new edge that is added to obtain the next tree in the listing:
\medskip

\begin{center}
\begin{tabular}{l@{\hspace{4em}}l@{\hspace{4em}}l}
     1. $-e_5+e_7$  &     \ 8. $-e_5+e_7$      &  \textcolor{red}{ 15. $\bm{-e_7+e_3}$} \\
     2. $-e_1+e_3$  &    \  9. $-e_4+e_3$   &  16. $-e_2+e_4$ \\
     3. $-e_3+e_4$  &   10. $-e_6+e_5$  &  17. $-e_1+e_2$ \\
     4. $-e_7+e_5$  &   \textcolor{red}{11.  $\bm{-e_5+e_2}$}     &  18. $-e_4+e_7$ \\
     5. $-e_4+e_3$  &  12.  $-e_2+e_4$  &  19. $-e_3+e_4$ \\
     6. $-e_2+e_1$  &   13. $-e_3+e_5$  &  20. $-e_5+e_3$.   \\
     7. $-e_3+e_4$  &   14. $-e_4+e_2$  & 
\end{tabular}
\end{center}

\smallskip

\noindent
This listing was generated from  Knuth's implementation of Smith's~\cite{smith1997generating} algorithm provided at \url{http://combos.org/span}. The steps in red are highlighted to show where the edge moves do not pivot around a vertex.
\end{exam}

A stronger notion of a Gray code for spanning trees is where the revolving-door makes strictly local changes.  More specifically, we would like the  edges being removed and added at each step to share a common endpoint.  We call a listing with this property a \textit{pivot Gray code} (also known as a \emph{strong revolving-door Gray code}~\cite{knuth}). The aforementioned spanning tree Gray codes are not pivot Gray codes.
In particular, the Gray code given by Smith's algorithm~\cite{smith1997generating} is not a pivot Gray code as illustrated in our previous example: the highlighted edge moves 11 and 15 do not have the ``pivot'' property.  
This leads to our first research question.

\begin{quote} \small
{\bf Research Question \#1} \ Given a graph $G$ (perhaps from a specific class), does there exist a pivot Gray code listing of all spanning trees of $G$? Furthermore, can the listing be generated in polynomial time per tree using polynomial space?
\end{quote}
A short discussion as to why previous methods do not lead directly to pivot Gray codes is presented in Section~\ref{sec:contract}.

A related question that arises for any listing is how to \emph{rank}, that is, find the position of the object in the listing,  and \emph{unrank}, that is, return the object at a specific rank. For spanning trees, an $O(n^3)$-time algorithm for ranking and unranking a spanning tree of a specific listing for an arbitrary graph is known~\cite{colbourn1989unranking}. 

\begin{quote} \small
{\bf Research Question \#2} \ Given a graph $G$ (perhaps from a specific class), does there exist a (pivot Gray code) listing of all spanning trees of $G$ that can be ranked and unranked in $O(n^2)$ time or better? 
\end{quote}

An algorithmic technique recently found to have success in the discovery of Gray codes is the greedy approach.  An algorithm is said to be \emph{greedy} if it can prioritize allowable actions according to some criteria, and then choose the highest priority action that results in a unique object to obtain the next object in the listing.
When applying a greedy algorithm, there is no backtracking; once none of the valid actions lead to a new object in the set under consideration, the algorithm halts, even if the listing is not exhaustive.  The work by Williams~\cite{williams2013greedy} notes that some very well-known combinatorial listings can be constructed greedily, including the binary reflected Gray code (BRGC) for binary strings, the plain change order for permutations,  and the lexicographically smallest de Bruijn sequence. 
Recently, a very powerful greedy algorithm on permutations (known as Algorithm J, where J stands for ``jump'') generalizes many known combinatorial Gray code listings including many related to permutation patterns, rectangulations, and elimination trees~\cite{MUTZE2020,MUTZEHoang2019,MUTZERectangulations2021}. However, no greedy algorithm was previously known to list the spanning trees of an arbitrary graph.

\begin{quote}  \small
{\bf Research Question \#3}  \ Given a graph $G$ (perhaps from a specific class), does there exist a greedy strategy to list all spanning trees of $G$?  Moreover, is the resulting listing a pivot Gray code? 
\end{quote}

\noindent
In most cases, a greedy algorithm requires exponential space to recall which objects have already been visited in a listing. Thus, answering this third question would satisfy only the first part of {\bf Research Question \#1}. However, in many cases, an underlying pattern can be found in a greedy listing which can result in space efficient algorithms~\cite{MUTZE2020,williams2013greedy}.  In recent communication with Arturo Merino, a greedy algorithm for listing the spanning trees of an arbitrary graph $G$ has been discovered by considering each tree's characteristic vector and transposing elements to change the
shortest possible prefix; however, it does not yield a pivot Gray code. 

To address these three research questions, we applied a variety of greedy approaches to structured classes of graphs including the fan, wheel, $n$-cube, and the compete graph.  From this study, we were able to affirmatively answer each of the research questions for the fan graph.  It remains an open question to find similar results for other classes of graphs.

\subsection{New results} \label{sec:results}

The \textit{fan graph} on $n$ vertices, denoted $F_n$, is obtained by joining a single vertex (which we label $v_\infty$) to the path on $n-1$ vertices (labeled $v_2, ... , v_n$) -- see Figure~\ref{fig:F5}. Note that we label 
\begin{wrapfigure}[7]{r}{0.35\textwidth}
\begin{center}
  \vspace*{-0.5cm}
  \hspace{-0.6cm}
  \includegraphics[scale=0.35, trim=0 0 0 0cm, clip]{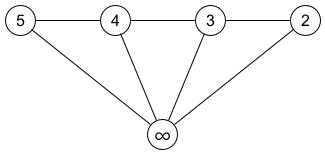}
  \caption{The fan $F_5$}
  \vspace{-0.2cm}
  \label{fig:F5}
\end{center}
\end{wrapfigure}
the smallest vertex $v_2$ so that the largest non-infinity labeled vertex equals the total number of vertices. 
We discover a greedy strategy to generate the spanning trees of $F_n$ in a pivot Gray code order. We describe this greedy strategy in Section~\ref{sec:greedy}. The resulting listing is studied to find an $O(1)$-amortized time recursive algorithm that produces the same listing using only $O(n)$ space, which is presented in Section~\ref{sec:recursion}. We also show how to rank and unrank a spanning tree of the greedy listing in $O(n)$ time in Section~\ref{sec:recursion}, which is a significant improvement over the general $O(n^3)$-time ranking and unranking that is already known.  The proofs for our main technical results are in Section~\ref{sec:proofs}. We conclude with a summary in Section~\ref{sec:summary}, along with a discussion as to how our pivot Gray code for the fan can be extended to the wheel.  

A complete C implementation of our algorithms is available in the Appendix.  A preliminary version of this paper appeared in COCOON 2021~\cite{confpaper}.


\subsection{Oriented spanning trees}

Although we only consider undirected graphs in this paper, we point out a related open problem for directed graphs.

Given a directed graph $D$ and a fixed root vertex $r$, an \emph{oriented spanning tree} or \emph{spanning arborescence} is an oriented subtree $T$ of $D$ with $n-1$ arcs such that there is a unique path from $r$ to every other vertex in $D$; all the arcs are directed away from $r$ in $T$.  The problem of finding a revolving-door Gray code for oriented spanning trees remains an open problem with a difficulty rating of 46/50 as given by Knuth in problem 102 on page 481 of~\cite{knuth}.  Knuth also notes on page 804 that a solution to this problem for a fixed root $r$ implies that a strong revolving-door (pivot) Gray code exists for the spanning trees of an undirected graph\footnote{The author's thank Torsten M\"{u}tze for pointing out this comment.}. The mapping here is natural: given an undirected graph $G$, replace all edges $(u,v)$ with two directed edges, one from $u$ to $v$ and one from $v$ to $u$. Algorithms to list all oriented spanning trees with a given root are known~\cite{gabow,Kapoor00analgorithm}; however, neither have the revolving-door property.


\subsection{Edge contraction and deletion} \label{sec:contract}
A technique applied in the construction of several ``revolving-door'' Gray codes~\cite{minty,smith1997generating,winter} is to recursively partition the spanning trees of a graph $G$ into those containing a specific edge $e$ by applying edge contraction, and those that do not contain the edge $e$ by deleting $e$.  However, when applying this strategy to construct a pivot Gray code, there are challenges when it comes to \emph{uncontracting an edge}.  Specifically, even if we have a pivot Gray code for the spanning trees in $G/ e$ ($G$ with the edge $e$ contracted), once we uncontract $e$, it does not necessarily result in a pivot Gray code for the original graph $G$.  See, for example Figure~\ref{fig:contract}.   

\begin{figure} [h]
\begin{center}
  \includegraphics[scale=0.70]{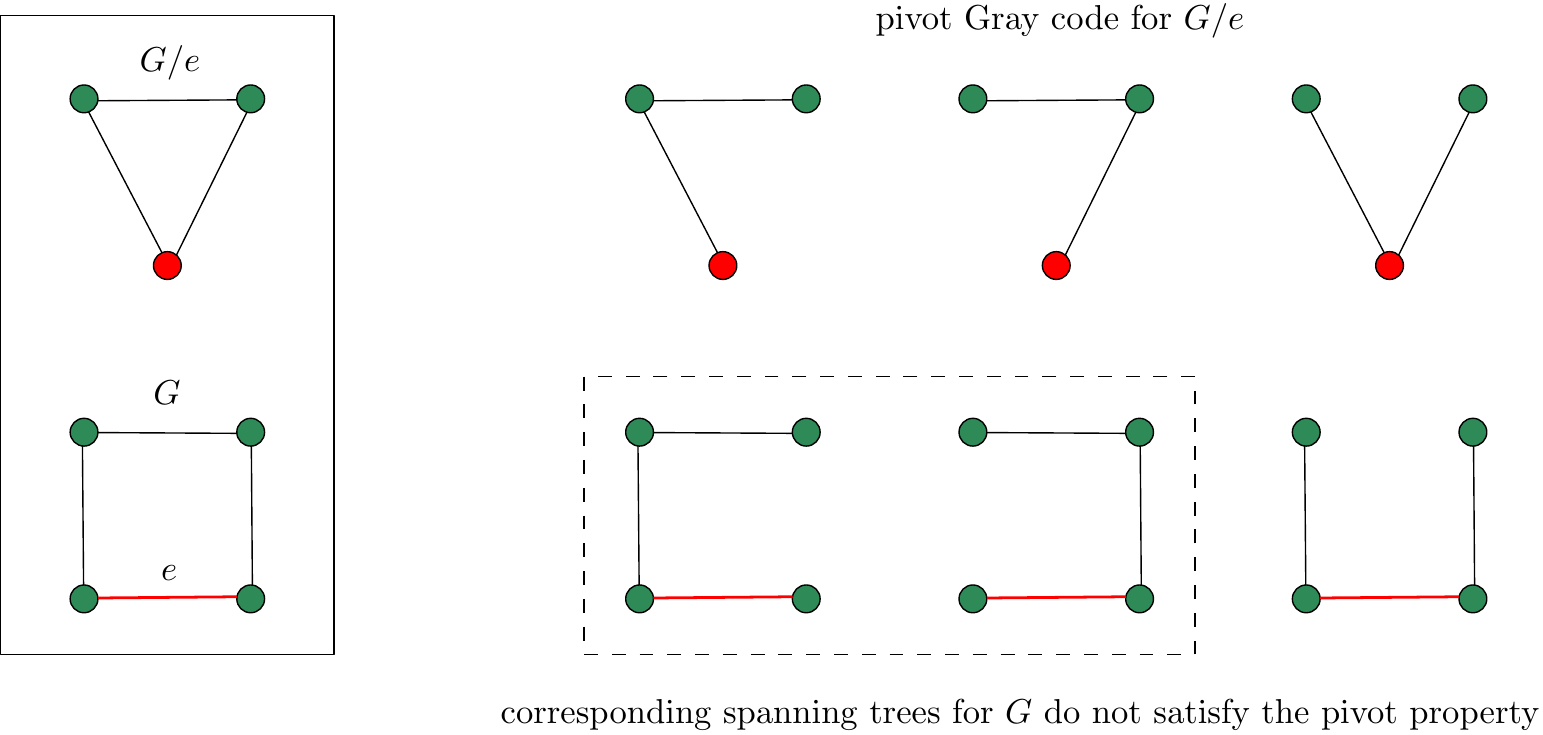}
  \caption{Illustrating how a pivot Gray code in graph $G/ e$ does not necessarily correspond to a pivot Gray code for the spanning trees of $G$ that include $e$.}
  \label{fig:contract}
\end{center}
\end{figure}


\section{A greedy approach}  \label{sec:greedy}

Every unicyclic graph has a cyclic pivot Gray code: start with any initial spanning tree then pivot the edges around the cycle one at time.  Thus, when considering greedy approaches for discovering a pivot Gray code we focused on the next obvious graph classes:  the complete graph $K_n$, the fan $F_n$, the wheel $W_n$, and the $n$-cube.   
There are two important issues when considering a greedy approach to list spanning trees: (1) the labels on the vertices (or edges) and (2) the starting tree. For each of our approaches, we prioritized our operations by first considering which vertex $u$ to pivot on, followed by an ordering of the endpoints considered in the addition/removal.  We call the vertex $u$ the \emph{pivot}.  

Our initial attempts focused only on pivots that were leaves.   As a specific example, we ordered the leaves (pivots) from smallest to largest.  Since each leaf $u$ is attached to a unique vertex $v$ in the current spanning tree, we then considered the neighbours of $u$ in increasing order of label.
We restricted the labeling of the vertices to the most natural ones, such as the one presented in Section~\ref{sec:results} for the fan.  For each strategy we tried all possible starting trees. Unfortunately, none of our attempts lead to exhaustive listings for  $K_n$, $F_n$, $W_n$, or the $n$-cube.

By allowing the pivot to be any arbitrary vertex, we ultimately discovered several exhaustive listing for the spanning trees of $F_n$; however, rather interestingly, we found no such listing for any other class.  
The listings we found for the fan were generated up to $n=12$. Starting from every starting tree for $n=12$ took about 8 hours on a single processor.
One listing stood out as having an easily defined starting tree as well as a nice pattern which we could study to construct the listing more efficiently.  It applied the labeling of the vertices as described in  Section~\ref{sec:results} with the following prioritization of pivots and their incident edges: 
\begin{quote}  
Prioritize the pivots $u$ from smallest to largest and then for each pivot, prioritize the edges $uv$ that can be removed from the current tree in increasing order of the label on $v$, and for each such $v$, 
prioritize the edges $uw$ that can be added to the current tree in increasing order of the label on $w$.
\end{quote}
Since this is a greedy strategy, if an edge pivot results in a spanning tree that has already been generated or a graph that is not a spanning tree, then the next highest priority edge pivot is attempted. Let \textsc{Greedy}$(T)$ denote the listing that results from applying this greedy approach starting with the spanning tree $T$. The starting tree that produced a nice exhaustive listing was the path $v_\infty, v_2, v_3, \ldots, v_n$, denoted $P_n$ throughout the paper. 
Figure~\ref{fig:F2F3F4F5} shows the listings  \textsc{Greedy}$(P_n)$ for $n=2,3,4,5$.  The listing  \textsc{Greedy}$(P_6)$ is illustrated in Figure~\ref{fig:F6Generation}.  It is worth noting that starting with the path $v_\infty,v_n, v_{n-1}, \ldots, v_2$ or the star (all edges incident to $v_\infty$) did not lead to an exhaustive listing for the spanning trees of $F_n$ in our study. 

\begin{figure}
\begin{center}
  \includegraphics[scale=0.235]{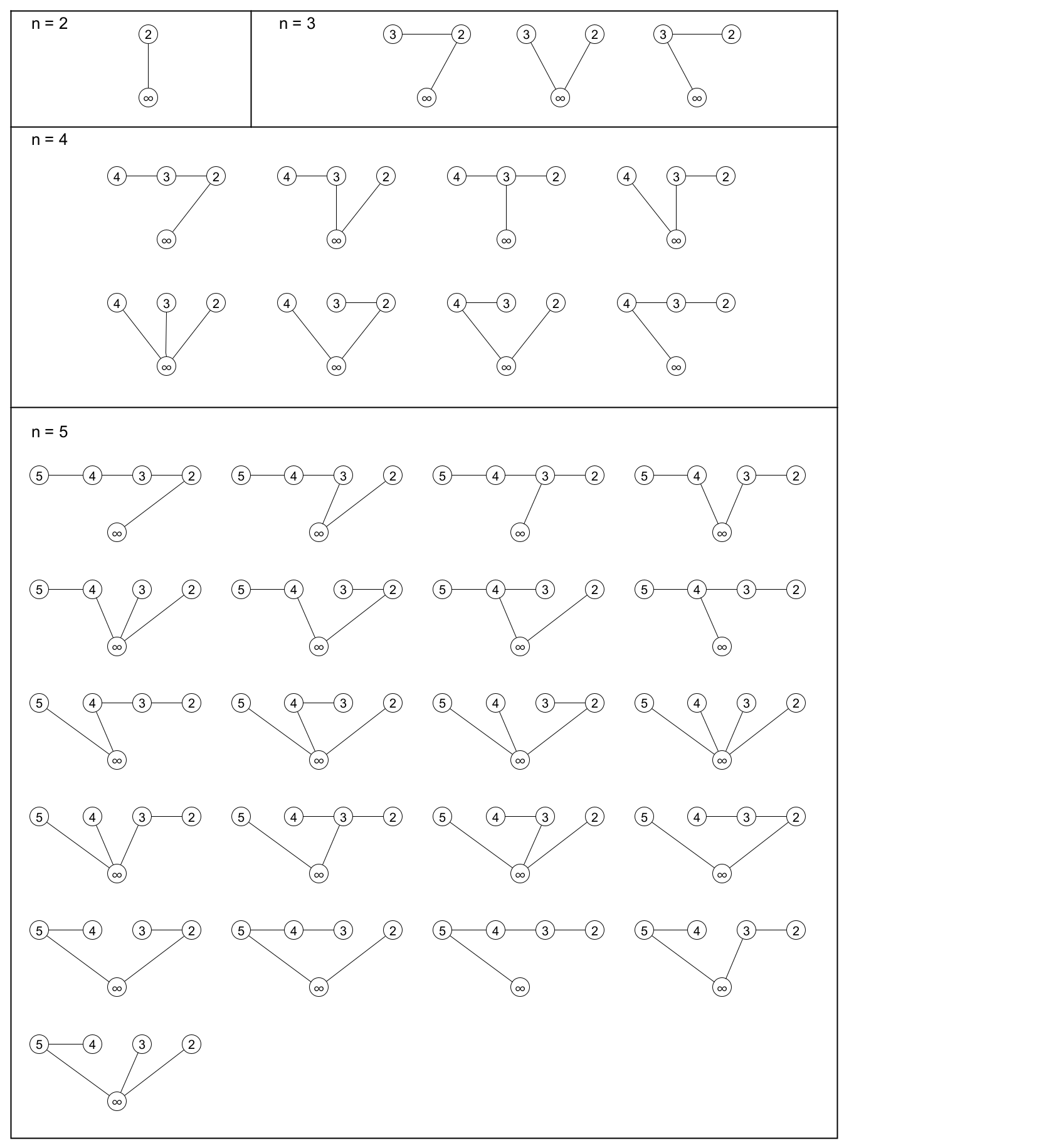}
  \caption{\textsc{Greedy}$(P_n)$ for $n=2,3,4,5$. Read left to right, top to bottom.}
  \label{fig:F2F3F4F5}
\end{center}
\end{figure}

\begin{figure}
\begin{center}
  \includegraphics[scale = 0.33]{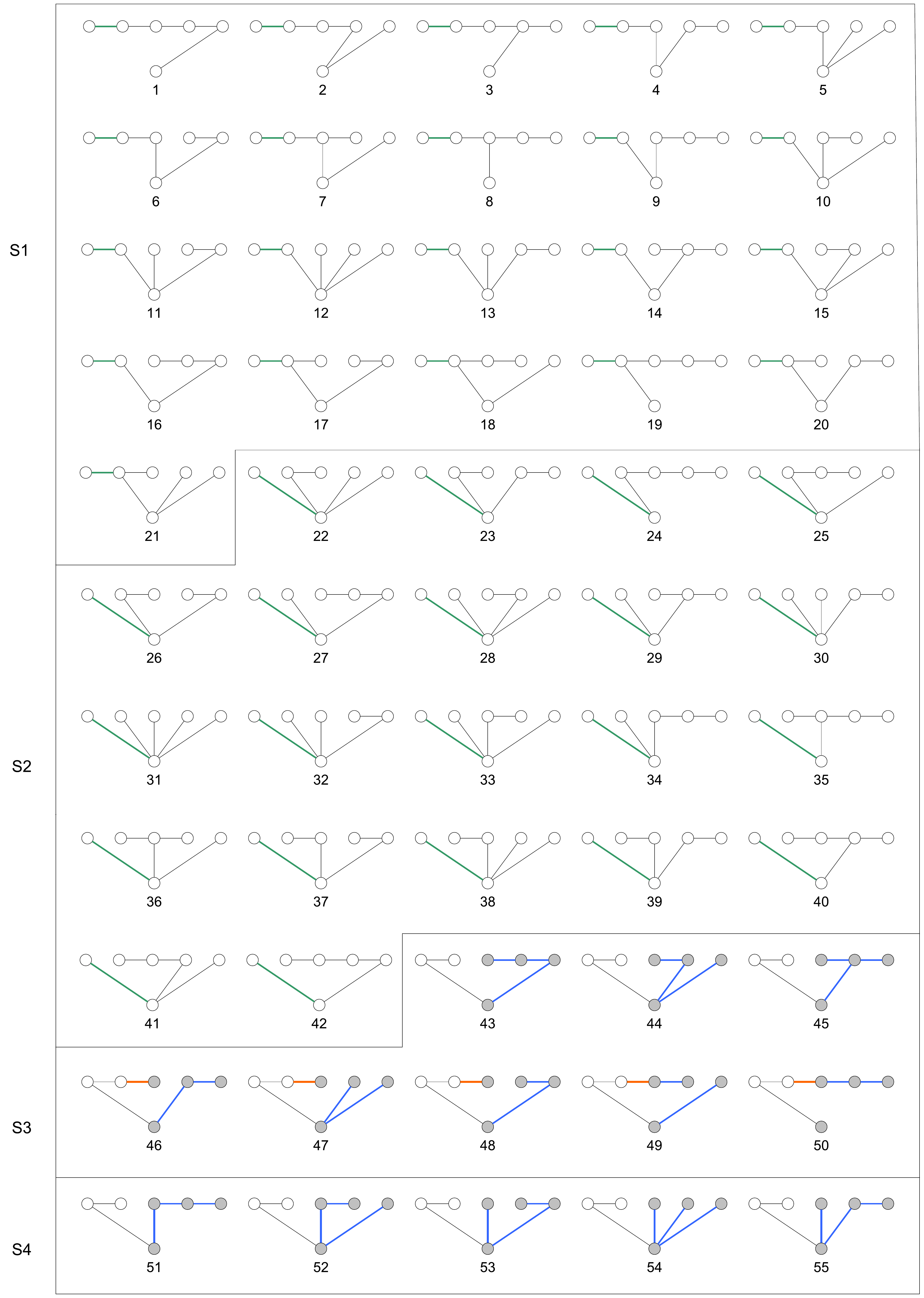}
  \caption{\textsc{Greedy}$(P_6)$ read from left to right, top to bottom. Observe that S1 is \textsc{Greedy}$(P_5)$ with $v_6 v_5$ added,  S2 is the reverse of \textsc{Greedy}$(P_5)$ with $v_6 v_\infty$ added,  S3 is \textsc{Greedy}$(P_4)$ with $v_6 v_5$ and $v_6 v_\infty$ added, except the edge $v_4 v_\infty$ is replaced by $v_4 v_5$, and S4 is the last five trees of \textsc{Greedy}$(P_4)$ in reverse order ($v_4 v_\infty$ is now present) with $v_6 v_5$ and $v_6 v_\infty$ added.}
  \label{fig:F6Generation}
\end{center}
\end{figure}

\begin{exam}
Consider the listing  \textsc{Greedy}$(P_5)$ in Figure~\ref{fig:F2F3F4F5}.  When the current tree $T$ is the 16th one in the listing (the one with edges $\{v_2v_\infty, v_2v_3, v_3v_4, v_5v_\infty\}$), the first pivot considered is $v_2$.  Since both $v_2v_3$ and $v_2v_\infty$ are present in the tree, no valid move is available by pivoting on $v_2$.  The next pivot considered is $v_3$.  Both edges $v_3v_2$ and $v_3v_4$ are incident with $v_3$. First, we attempt to remove $v_3v_2$ and add $v_3v_\infty$, which results in a tree previously generated. Next, we attempt to remove $v_3v_4$ and add $v_3v_\infty$, which results in a cycle. So, the next pivot, $v_4$, is considered.  The only edge incident to $v_4$ is $v_4v_3$.  By removing $v_4v_3$ and adding $v_4v_5$ we obtain a new spanning tree, the next tree in the greedy listing.
\end{exam}

 To prove that \textsc{Greedy}$(P_n)$ does in fact contain all the spanning trees of $F_n$, the next section
demonstrates it is equivalent to a recursively constructed listing obtained by studying the greedy listings.  Before we describe this recursive construction we mention one rather remarkable property of \textsc{Greedy}$(P_n)$ that we prove later in Section~\ref{sec:proofs}.

\begin{remark}
Let $X_n$ be last tree in the listing \textsc{Greedy}$(P_n)$. Then \textsc{Greedy}$(X_n)$ is precisely \textsc{Greedy}$(P_n)$ in reverse order.
\end{remark}


\section{A pivot Gray code for the spanning trees of $F_n$ via recursion}  \label{sec:recursion}

In this section we develop an efficient recursive algorithm to construct the listing  \textsc{Greedy}$(P_n)$.  The construction generates some sub-lists in reverse order, similar to the recursive construction of the BRGC.  The recursive properties allow us to provide efficient ranking and unranking algorithms for the listing based on counting the number of trees at each stage of the construction.   Let \spt{n} denote the number of spanning trees of $F_n$. It is known that 
$$\spt{n} = f_{2(n-1)} =  2 \frac{((3-\sqrt{5})/2)^{n}-((3+\sqrt{5})/2)^{n-2}}{5-3\sqrt{5}},$$ 
where $f_n$ is the $n$th number of the Fibonacci sequence with $f_1=f_2=1$ \cite{fanformula}.

\subsection{Pivot Gray code construction}

By studying the order of the spanning trees in \textsc{Greedy}$(P_n)$,
we identified four distinct stages S1, S2, S3, S4 that are highlighted for \textsc{Greedy}$(P_6)$ in 
Figure~\ref{fig:F6Generation}.  From this figure, and referring back to Figure~\ref{fig:F2F3F4F5} to see the recursive properties, observe that:
\begin{itemize}
    \item The trees in S1 are  equivalent to \textsc{Greedy}$(P_{5})$ with the added edge $v_6 v_{5}$. 
    
    \item The trees in S2 are equivalent to the reversal of the trees in \textsc{Greedy}$(P_{5})$ with the added edge $v_6 v_{\infty}$. 
\end{itemize}
\noindent
The trees in S3 and S4 have both edges $v_6 v_{5}$ and $v_6 v_{\infty}$ present.  
\begin{itemize}

    \item In S3,  focusing only on the vertices $v_4,v_3,v_2,v_\infty$, the induced subgraphs correspond to  \textsc{Greedy}$(P_{4})$, except whenever  $v_4v_\infty$ is present, it is replaced with $v_4v_5$ (the last five trees).
    \item  In S4, focusing only on the vertices $v_4,v_3,v_2,v_\infty$, the induced subgraphs correspond to the trees in \textsc{Greedy}$(P_{4})$
    where $v_4v_\infty$ is  present, in reverse order.
\end{itemize} 

Generalizing these observations for all $n \geq 2$ leads to the recursive procedure given in Algorithm~\ref{alg: Gen}, called {\sc Gen}($k,s_1,\mathit{varEdge}$). It uses a global variable $T$ to store the current spanning tree with $n$ vertices. The parameter $k$ indicates the number of vertices under consideration; the parameter $s_1$ indicates whether or not to generate the trees in stage S1, as required by the trees for S4; and the parameter $\mathit{varEdge}$ indicates whether or not a variable edge needs to be added as required by the trees for S3.   The base cases correspond to the edge moves in the listings \textsc{Greedy}$(P_2)$ and \textsc{Greedy}$(P_3)$.  

\begin{result} \noindent
Let $\LIST{n}$ denote the listing obtained by initializing $T$ to $P_n$, printing $T$, and  calling {\sc Gen}($n,1,0$).
\end{result}

\begin{algorithm}[t]
    \footnotesize \caption{}
    \label{alg: Gen}
	\begin{algorithmic}[1]
	    \Procedure{Gen}{$k, s_1, varEdge$}
		\If{$k = 2$} \Comment{$F_2$ base case}
			\If{$varEdge$} $T \gets T - v_2 v_\infty + v_2 v_3$; \textsc{Print}$(T)$
			\EndIf
		\ElsIf{$k = 3$} \Comment{$F_3$ base case}
			 \If{$s_1$} 
				\If{$varEdge$} $T \gets T - v_3 v_2 + v_3 v_4$; \textsc{Print}$(T)$
				\Else{ $T \gets T - v_3 v_2 + v_3 v_\infty$}; \textsc{Print}$(T)$
				\EndIf
			\EndIf	
		 \State $T \gets T - v_2 v_\infty + v_2 v_3$; \textsc{Print}$(T)$
		\Else 
			\If{$s_1$} 
				\State \Gen{}$(k-1, 1, 0)$	 \Comment{S1}
				\If{$varEdge$} $T \gets T - v_k v_{k-1} + v_k v_{k+1}$; \textsc{Print}$(T)$
				\Else{ $T \gets T - v_k v_{k-1} + v_k v_\infty$}; \textsc{Print}$(T)$
				\EndIf
			\EndIf 
			\State \RevGen{}$(k-1, 1, 0)$ \Comment{S2}
			\State $T \gets T - v_{k-1} v_{k-2} + v_{k-1} v_k$; \textsc{Print}$(T)$
			\State	\Gen{}$(k-2, 1, 1)$  \Comment{S3}
			\If{$k > 4$} $T \gets T - v_{k-2} v_{k-1} + v_{k-2} v_\infty$; \textsc{Print}$(T)$
			\EndIf 
			\State \RevGen{}$(k-2, 0, 0)$  \Comment{S4}
		\EndIf
		\EndProcedure
	\end{algorithmic}
\end{algorithm}

Before discussing {\sc RevGen},
we first provide a formal description of the last tree in the listing $\LIST{n}$, which we denote $L_n$.
Define the tree $Last_n$ as follows for $n \geq 2$: for $n=2,3,4$ let $Last_n$ be the last trees in the listings for $n=2,3,4$ given in Figure~\ref{fig:F2F3F4F5}, and for $n > 4$ let 
\begin{equation*}
    Last_n = Last_{n-3} + v_n v_{n-1} + v_n v_\infty + v_{n-2} v_\infty.
\end{equation*}
Applying this definition, the trees $Last_n$ for $2\leq n \leq 7$ are given in Figure~\ref{fig:LastTrees}.  The following lemma is proved in Section~\ref{sec:proofs}.

\begin{lemma} \label{lem:last}
For $n \geq 2$, $L_n = Last_n$.
\end{lemma}
 
\begin{figure}
\begin{center} 
  \includegraphics[scale = 0.27]{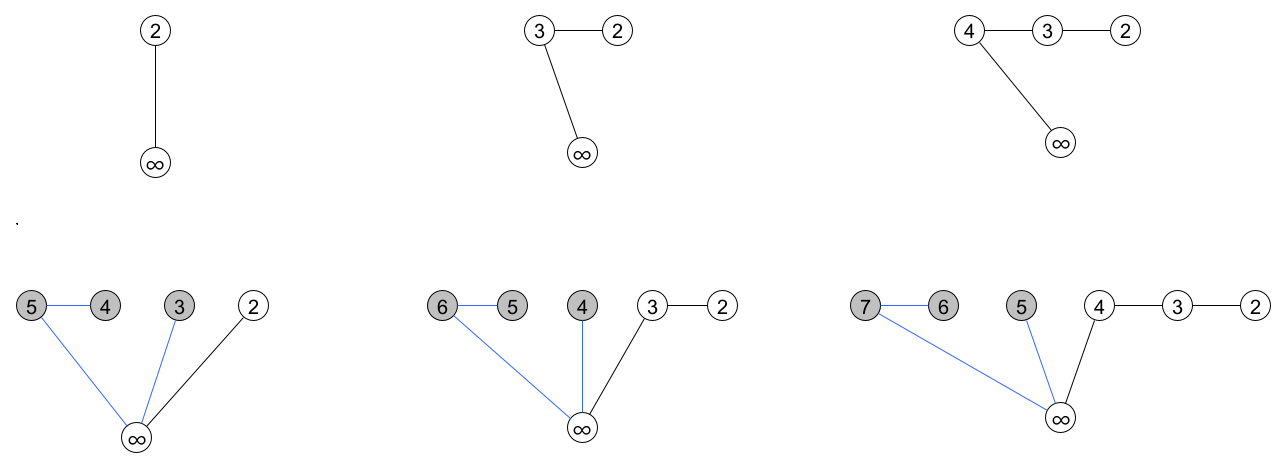}
  \caption{From left to right. \{$Last_2$, $Last_3$, $Last_4$\} (top), \{$Last_5$, $Last_6$, $Last_7$\} (bottom). Shaded vertices and blue edges highlight the additional vertices and edges added to the trees in the top row to obtain the trees in the bottom row. }
  \vspace{-1cm}
  \label{fig:LastTrees}
\end{center}
\end{figure}

The procedure {\sc RevGen}($k,s_1,\mathit{varEdge}$),  performs the operations of {\sc Gen}($k,s_1,\mathit{varEdge}$) in reverse order, thus producing the reversal of the listing generated by {\sc Gen}($k,s_1,\mathit{varEdge}$) when starting with the last tree from the latter listing.  The base cases correspond to the edge moves in the reversals of the listings \textsc{Greedy}$(P_2)$ and \textsc{Greedy}$(P_3)$. 

\begin{result}
\noindent
Let $\REVLIST{n}$ denote the listing obtained by initializing $T$ to $L_n$, printing $T$, and  calling {\sc RevGen}($n,1,0$). 
\end{result}

\begin{remark}
$\REVLIST{n}$ is the listing $\LIST{n}$ in reverse order. 
\end{remark}

\begin{algorithm}[t]
    \footnotesize \caption{}\label{alg: RevGen}
	
	\begin{algorithmic}[1]
	   \Procedure{RevGen}{$k, s_1, varEdge$}
		\If{$k = 2$} \Comment{$F_2$ base case}
			\If{$varEdge$} $T \gets T - v_2 v_3 + v_2 v_\infty$; \textsc{Print}$(T)$
			\EndIf
		\ElsIf{$k = 3$} \Comment{$F_3$ base case}
			\State $T \gets T - v_2 v_3 + v_2 v_\infty$; \textsc{Print}$(T)$
			 \If{$s_1$} 
				\If{$varEdge$}  $T \gets T - v_3 v_4 + v_3 v_2$; \textsc{Print}$(T)$
				\Else{ $T \gets T - v_3 v_\infty + v_3 v_2$}; \textsc{Print}$(T)$
				\EndIf
			\EndIf	
		\Else 
			\State \Gen{}$(k-2, 0, 0)$  \Comment{S4}
			\If{$k > 4$} $T \gets T - v_{k-2} v_\infty + v_{k-2} v_{k-1}$; \textsc{Print}$(T)$
			\EndIf 	
			\State	\RevGen{}$(k-2, 1, 1)$  \Comment{S3}
			\State $T \gets T - v_{k-1} v_k + v_{k-1} v_{k-2}$; \textsc{Print}$(T)$
			\State \Gen{}$(k-1, 1, 0)$ \Comment{S2}
			\If{$s_1$} 
				\If{$varEdge$} $T \gets T - v_k v_{k+1} + v_k v_{k-1}$; \textsc{Print}$(T)$
				\Else{ $T \gets T - v_k v_\infty + v_k v_{k-1}$}; \textsc{Print}$(T)$
				\EndIf
				\State \RevGen{}$(k-1, 1, 0)$	 \Comment{S1}
			\EndIf 
		\EndIf
	    \EndProcedure
	\end{algorithmic}
\end{algorithm}


\begin{theorem} \label{thm:main}
For $n\geq 2$, $\LIST{n}$ and $\REVLIST{n}$ are pivot Gray codes for the spanning trees of the fan $F_n$ and they can be generated in $O(1)$-amortized time using $O(n)$ space.  Moreover, \textsc{Greedy}$(P_n)$ = $\LIST{n}$ and \textsc{Greedy}$(L_n)$ = $\REVLIST{n}$.
\end{theorem}
We prove this theorem in Section~\ref{sec:proofs}.   

\subsection{Ranking}  \label{sec:ranking}


Given a spanning tree $T$ in $\LIST{n}$, we calculate its rank by recursively determining which stage (recursive call) $T$ is generated. We can determine the stage by focusing on the presence/absence of the edges $v_n v_{n-1}$, $v_n v_\infty$, $v_{n-2} v_\infty$, and $v_{n-2} v_{n-1}$.  Based on the discussion of the recursive algorithm, there are $t_{n-1}$ trees generated in S1,  $t_{n-1}$ trees generated in S2,  $t_{n-2}$ trees generated in S3, and $t_{n-2} - t_{n-3}$ trees generated in S4. S3 is partitioned into two cases based on whether $v_{n-2} v_{n-1}$ ($varEdge)$ is present. For the remainder of this section we will let $T_{n-1} = T - v_n$ and $T_{n-2} = T - v_n - v_{n-1}$.

If $v_n v_{n-1}, v_n v_\infty, v_{n-2} v_\infty \in T$, then $T$ is a tree in S4 of $\LIST{n}$. The trees of S4 are the trees of $\LIST{n-2}$ without S1, listed in reverse order. So, the rank can be calculated by subtracting the rank of $T_{n-2}$ in $\LIST{n-2}$ from $2\spt{n-1} + 2\spt{n-2}$ (the rank of the last tree of S3 plus $\spt{n-2}$). Note that we do not use $2 \spt{n-1} + 2 \spt{n-2} - \spt{n-3}$ because the recursive rank calculated already takes into account the trees of $S1$ that are missing.

If $v_n v_{n-1}, v_n v_\infty \in T$ and $v_{n-2} v_\infty \not \in T$, then $T$ is a tree in S3 of $\LIST{n}$. The trees of S3 are the trees of $\LIST{n-2}$ where $v_{n-2} v_\infty$ has been replaced by $v_{n-2} v_{n-1}$. So, if $v_{n-2} v_{n-1} \in T$, then in order to recursively calculate the rank of $T_{n-2}$ in $\LIST{n-2}$, we need to replace $v_{n-2} v_{n-1}$ with $v_{n-2} v_\infty$. If $v_{n-2} v_{n-1} \not \in T$, then no edge replacements are needed. We can then determine the rank of $T$ in $\LIST{n}$ by adding the rank of $T_{n-2}$ in $\LIST{n-2}$ to $2 \spt{n-1}$ (the rank of the last tree of S2).

The other two cases are fairly trivial. If $v_n v_{n-1} \in T$ and $v_n v_\infty \not \in T$, then $T$ is in S1. Since S1 is the trees of $\LIST{n-1}$ with $v_n v_{n-1}$ added, we simply return the ranking of $T_{n-1}$ in $\LIST{n-1}$. If $v_n v_\infty \in T$ and $v_n v_{n-1} \not \in T$, then $T$ is in S2. Since S2 is the trees of $\LIST{n-1}$ in reverse order with $v_n v_\infty$ added, we return $2  \spt{n-1} + 1$ (the rank of the first tree of S3) minus the rank of $T_{n-1}$ in $\LIST{n-1}$.

For $n>1$, let $R_n(T)$ denote the rank of $T$ in the listing $\LIST{n}$. If $n=2,3,4$, then $R_n(T)$ can easily be derived from Figure~\ref{fig:F2F3F4F5}.  Based on the above discussion, for $n\geq 5$: 
\begin{small}
\begin{equation} \label{eq:rank}
 R_n(T) = 
       \begin{cases}
		2 \spt{n-1} + 2 \spt{n-2} - R_{n-2}(T_{n-2}) + 1 
		& \text{if $e_1, e_2, e_3 \in T$} \\ 

		2 \spt{n-1} + R_{n-2}(T_{n-2} + e_3) 
		& \text{if $e_1, e_2, e_4 \in T$, $e_3 \not \in T$} \\

		2 \spt{n-1} + R_{n-2}(T_{n-2})
		& \text{if $e_1, e_2 \in T$, $e_3, e_4 \not \in T$} \\ 

		2 \spt{n-1} - R_{n-1}(T_{n-1}) + 1 & \text{if $e_2 \in T$, $e_1 \not \in T$} \\ 
		R_{n-1}(T_{n-1}) & \text{if $e_1 \in T$, $e_2 \not \in T$}
        \end{cases} 
\end{equation}
\end{small}
where $e_1 = v_n v_{n-1}$, $e_2 = v_n v_\infty$, $e_3 = v_{n-2} v_\infty$, and $e_4 = v_{n-2} v_{n-1}$.

\medskip

\begin{exam} 
Consider the spanning trees $T_7$, $T_6$, $T_5$, and $T_3$ for $F_7, F_6, F_5$ and $F_3$ respectively.   

\begin{center}
  \includegraphics[scale = 0.16]{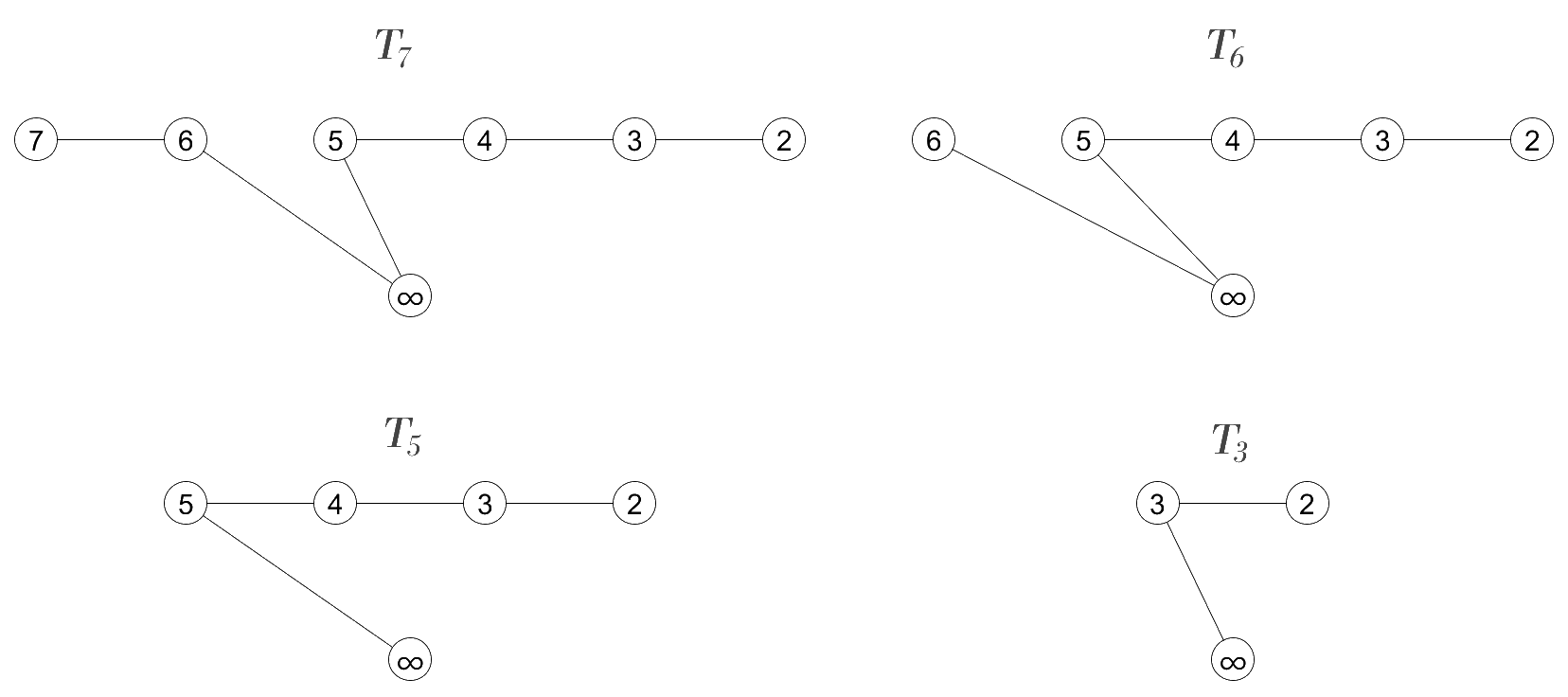}
  \label{fig:RankingExample}
\end{center}

\noindent Observe that
$T_6 = T_7 - v_7$, 
$T_5 = T_6 - v_6$, and
$T_3 = T_5 - v_5 - v_4 + v_3v_\infty$.
Consider $R_7(T)$ where $T= T_7$.  Applying the formula in (\ref{eq:rank})  we have
\begin{equation*}
    \begin{split}
        R_7(T)   & = R_{6}(T_6) \\
                    & = 2 \spt{5} - R_{5}(T_5) + 1\\
                    & = 43 - (2 \spt{4} + R_3(T_3 + v_3 v_\infty)) \\
                    & = 43 - (16 + 3) \\
                    & = 24.
    \end{split}
\end{equation*}

\end{exam}

Since each application of (\ref{eq:rank}) requires constant time, and the recursion goes $O(n)$ levels deep,  we arrive at the following result provided the first $2(n{-}2)$ Fibonacci numbers are precomputed. We note that the calculations are on numbers up to size $t_{n-1}$.

 \begin{theorem} \label{thm:rank}
 The listing $\LIST{n}$  can be ranked in $O(n)$ time using $O(n)$ space under the unit cost RAM model.
 \end{theorem}
 %


\subsection{Unranking}

Determining the tree $T$ at rank $r$ in the listing $\LIST{n}$ follows similar ideas by constructing $T$ starting from a set of $n$ isolated vertices and adding one edge at a time. If $0 < r \leq \spt{n-1}$ then $T$ must be a tree in S1 of $\LIST{n}$. So, we can  add $v_n v_{n-1}$ to $T$ and consider the rank $r$ tree in $\LIST{n-1}$. If $\spt{n-1} < r \leq 2 \spt{n-1}$, then $T$ is a tree in S2 of $\LIST{n}$. Since $S2$ of $\LIST{n}$ is simply $\REVLIST{n-1} + v_n v_\infty$, we can add $v_n v_\infty$ to $T$ and then consider the rank $2 \spt{n-1} + 1 - r$ (rank of the first tree of S3 minus $r$) tree in $\LIST{n-1}$. If $2 \spt{n-1} < r \leq 2 \spt{n-1} + \spt{n-2}$, then $T$ must be a tree in S3 of $\LIST{n}$. Since all trees in S3 have the edges $v_n v_{n-1}$ and $v_n v_\infty$, we can add these edges to $T$. Then, we can consider the rank $r - 2\spt{n-1}$ ($r$ minus the rank of the last tree of S2) tree of $\LIST{n-2}$. Also note that since $T$ is in S3, $v_{n-2} v_{n-1}$ will replace $v_{n-2} v_\infty$ for the trees of $\LIST{n-2}$. Otherwise, if $r > 2 \spt{n-1} + \spt{n-2}$, then $T$ must be in S4 of $\LIST{n-2}$. Similar to S3, we can add $v_n v_{n-1}$ and $v_n v_\infty$ to $T$ as all trees in S4 have these edges. Then, we consider the $2 \spt{n-1} + 2 \spt{n-2} - r + 1$ (rank of the last tree of S3 plus the rank of the last tree of $\LIST{n-2}$ minus $r$) tree of $\LIST{n-2}$.

Let $U_n(r, \mathit{replaceEdge})$ return the edges that form the tree $T$ at rank $r$ for the listing $\LIST{n}$. The parameter $\mathit{replaceEdge}$ indicates whether or not the edge $v_n v_{n+1}$ should  be added instead of $v_n v_\infty$. Initially, $r$ is the specified rank, and $\mathit{replaceEdge} = 0$.
In the base cases where $n=2,3,4$, then $T$ is  derived from Figure~\ref{fig:F2F3F4F5}. For these cases, if the edge $v_n v_\infty$ is present and $\mathit{replaceEdge} = 1$, then it is replaced by the edge $v_n v_{n+1}$. Based on the above discussion, we arrive at the following recursive construction for $U_n(r, \mathit{replaceEdge})$.

\begin{footnotesize}
\begin{equation} \label{eq:unrank} 
 U_n(r, \mathit{replaceEdge}) = 
       \begin{cases}
       	U_{n-1}(r, 0) + v_n v_{n-1}
    	& \text{if $0 < r \leq \spt{n-1}$,} \\ 
    
    	U_{n-1}(2 \spt{n-1} {-} r {+} 1, 0) + e
    	& \text{if $\spt{n-1} < r \leq 2 \spt{n-1}$,} \\ 
    
    	U_{n-2}(r {-} 2\spt{n-1}, 1) + v_n v_{n-1} + e
    	& \text{if $2 \spt{n-1} < r \leq 2 \spt{n-1} {+} \spt{n-2}$,} \\ 
    
    	U_{n-2}(2 \spt{n-1} {+} 2 \spt{n-2} {-} r + 1, 0) + v_n v_{n-1} + e
    	& \text{otherwise,}



        \end{cases}  
\end{equation}
\end{footnotesize}
where $e = v_n v_{n+1}$ if $\mathit{replaceEdge} = 1$ and $e = v_n v_\infty$ otherwise. 

\medskip
\begin{exam} \small
To find the 24th tree $T$ in the listing $\LIST{7}$, we consider $U_7(24, 0)$.  Repeated application of (\ref{eq:unrank}) yields the following 
\begin{eqnarray*}
U_7(24, 0)  &  =  & U_6(24, 0) +  v_7 v_6 \ \ \ \ \  \mbox{\blue{$\triangleright$ since $0 < 24 \leq \spt{6}$}} \\
                &  =  & U_5(19, 0) +  v_7 v_6 + v_6 v_\infty   \ \ \ \ \  \mbox{\blue{$\triangleright$  since $\spt{5}  < 24 \leq 2 \spt{5}$}} \\
                & = &  U_3(3, 1) +  v_7 v_6 + v_6 v_\infty + v_5 v_4 + v_5 v_\infty  \ \ \ \ \ \mbox{\blue{$\triangleright$  since $2 \spt{4}  < 19 \leq 2 \spt{4} + \spt{3} $}} \\
                & = &  v_7 v_6 + v_6 v_\infty + v_5 v_4 + v_5 v_\infty + v_3 v_3 + v_3 v_2 \ \ \ \mbox{} 
\end{eqnarray*}
Reaching a base case, the 3rd tree of $\LIST{3}$ is $\{ v_3 v_2, v_3 v_\infty \}$. Since $\mathit{replaceEdge} = 1$, the edge $v_3v_\infty$ is replaced with $v_3v_4$ and we end up with the spanning tree $T$ containing the edges from the last line of the equation.  These four steps to construct $T$ are illustrated below. 
%

\begin{center}
  \includegraphics[scale = 0.15]{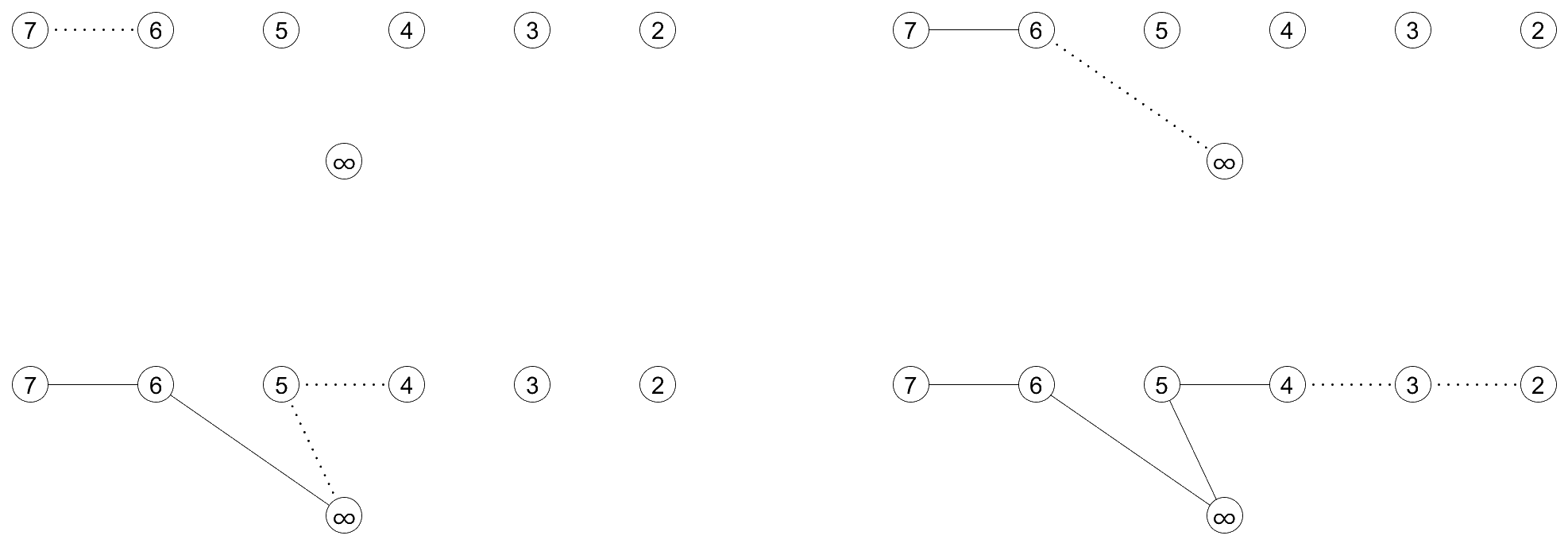}
  \label{fig:UnrankingExample}
\end{center}
\end{exam}

Since each application of (\ref{eq:unrank}) requires constant time, and the recursion goes $O(n)$ levels deep,  we arrive at the following result provided the first $2(n{-}2)$ Fibonacci numbers are precomputed. We note that the calculations are on numbers up to size $t_{n-1}$.

 \begin{theorem} \label{thm:unrank}
 The listing $\LIST{n}$  can be unranked in $O(n)$ time using $O(n)$ space under the unit cost RAM model.
 \end{theorem}
 

\section{Proof of technical results} \label{sec:proofs}

\subsection{Proof of Theorem~\ref{thm:main}}

To prove Theorem~\ref{thm:main}, we start by proving that the number of trees generated by $\LIST{n}$ is $t_n$.  Then we show that  $\LIST{n}$ = {\sc Greedy}($P_n$) and $\REVLIST{n}$ = {\sc Greedy}($L_n$).  Combining these results with the fact that the trees generated by the greedy approaches are unique and successive trees differ by the ``pivot'' of a single edge,  we have that $\LIST{n}$ $\REVLIST{n}$ are pivot Gray codes for the spanning trees of the fan graph $F_n$.  Finally, we verify the running time of the recursive algorithm to generate $\LIST{n}$.


Before proving these results, we introduce some notation.
Let $T - v_i$ denote the tree obtained from $T$ by deleting the vertex $v_i$  along with all edges that have $v_i$ as an endpoint. Let $T + v_i v_j$ (resp. $T-v_i v_j$) denote the tree obtained from $T$ by adding (resp. deleting) the edge $v_i v_j$. For the remainder of this section, we will let $T_n$ denote the tree $T$ specified as a global variable for \Gen{} and \RevGen{}, and we let $T_{n-1}=T-v_n$ and $T_{n-2}=T-v_n-v_{n-1}$. 

\begin{lemma} \label{EdgeMovesLemma}
For $n \geq 2$, $|\LIST{n}| = |\REVLIST{n}| = \spt{n}$.
\end{lemma}

\begin{proof}
We first note that since \Gen{} and \RevGen{} are exact reversals of each other, \Gen{}$(n, s_1, varEdge)$ starting with $T=P_n$ and \RevGen{}$(n, s_1, varEdge)$ starting with $T=L_n$ produce the same number of trees. The proof now proceeds by induction on $n$.  It is easy to verify the result holds for $n=2,3,4$. Now assume $n>4$, and that $|\LIST{j}| = \spt{j}$, for $2 \leq j<n$.  We consider the number of trees generated by each of the four stages of \Gen{}$(n-1, 1, 0)$ when starting with $P_n$.\\

\noindent \underline{S1:} Since $n>4$ and $s_1 = 1$, \Gen{}$(n-1, 1, 0)$  is executed. Since $T_n = P_n$, we have that $T_{n-1} = P_{n-1}$. So, by our inductive hypothesis, $\spt{n-1}$ trees are printed. By definition of $L_n$, $T_{n-1} = L_{n-1}$ after \Gen{}$(n-1, 1, 0)$. It follows that $T_n = L_{n-1} + v_n v_{n-1}$. Line 12 removes $v_n v_{n-1}$ and adds $v_n v_\infty$. Since $v_n v_{n-1} \in T$ and $v_n v_\infty \not \in T$, this results in one more tree printed. At this point, $T_n = L_{n-1} + v_n v_\infty$. \\

\noindent \underline{S2:} Next, line 14 executes \RevGen{}$(n-1, 1, 0)$. We have that $T_{n-1} = L_{n-1}$. So, by the inductive hypothesis, $\spt{n-1}$ trees are printed. We know that $T_{n-1} = P_{n-1}$ after \RevGen{}$(n-1, 1, 0)$ starting with $T_{n-1} = L_{n-1}$, so it follows that $T_n = P_{n-1} + v_n v_\infty$. Line 15 removes $v_{n-2} v_{n-1}$ and adds $v_{n-1} v_n$. Since $v_{n-2} v_{n-1} \in T$ (because $v_{n-2} v_{n-1} \in P_{n-1}$) and $v_{n-1} v_n \not \in T$, this results in one more tree printed. At this point, $T_n = P_{n-2} + v_n v_\infty + v_{n-1} v_n$.\\

\noindent \underline{S3:} Line 16 then executes \Gen{}$(n-2, 1, 1)$ with $T_{n-2} = P_{n-2}$ since $T_n = P_{n-2} + v_n v_\infty + v_{n-1} v_n$. Note that the only difference between \Gen{}$(n-2, 1, 1)$ and \Gen{}$(n-2, 1, 0)$ is that $v_j v_{j+1}$ is added instead of $v_j v_\infty$ since $n-2>2$. Also, $v_{n-2} v_{(n-2)+1} \not \in T_n$ so it can be added. It follows that \Gen{}$(n-2, 1, 1)$ and \Gen{}$(n-2, 1, 0)$ will output the same number of trees starting with $T_{n-2} = P_{n-2}$. So, line 16 results in $\spt{n-2}$ trees printed, again by the inductive hypothesis. After line 16 is executed, we have $T_{n-2}=L_{n-2} - v_{n-2} v_\infty + v_{n-2} v_{n-1}$ since $varEdge$ was equal to $1$. Line 17 removes $v_{n-2} v_{n-1}$ and adds $v_{n-2} v_\infty$. Since $v_{n-2} v_{n-1} \in T$ and $v_{n-2} v_\infty \not \in T$, this results in one more tree printed. At this point, $T_n = L_{n-2} + v_n v_{n-1} + v_n v_\infty$.\\

\noindent \underline{S4:} By our inductive hypothesis, $|\REVLIST{n-2}| = \spt{n-2}$. However, $s_1 = 0$ for line 18 (\RevGen{}$(n-2, 0, 0)$). So, for \RevGen{}$(n-2, 0, 0)$, line 17 (one tree) and line 18 ($\spt{n-3}$ trees) are not executed. This results in a total of $\spt{n-2} - \spt{n-3}$ trees being printed by line 18 of \Gen{}$(n, 1, 0)$. 

In total, $2 \spt{n-1} + 2 \spt{n-2} - \spt{n-3}$ trees are printed. By a straightforward Fibonacci identity which we leave to the reader, we have that $\spt{n} = 2 \spt{n-1} + 2 \spt{n-2} - \spt{n-3}$. Therefore, $|\LIST{n}| = |\REVLIST{n}| = \spt{n}$. 
\end{proof}

To prove the next result, we first detail some required terminology.
If $T$ is a spanning tree of $F_n$, then we say that the operation of deleting an edge $v_i v_j$ and adding an edge $v_i v_k$ is a \emph{valid} edge move of $T$ if the result is a spanning tree that has not been generated yet. Conversely, if the result is not a spanning tree, or the result is a tree that has already been generated, then it is not a \emph{valid} edge move of $T$. We say an edge $v_i v_j$ is \emph{smaller} than edge $v_i v_k$ if $j<k$. An edge move $T_n - v_i v_j + v_i v_k$ is said to be \emph{smaller} than another edge move $T_n - v_x v_y + v_x v_z$ if $i<x$, if $i=x$ and $j<y$, or if $i=x$, $j=y$, and $k<z$.

\begin{lemma} \label{Gen=GreedyTheorem}  
For $n\geq 2$,  $\LIST{n} = \textsc{Greedy}(P_n)$ and $\REVLIST{n} = \textsc{Greedy}(L_n)$.
\end{lemma}

\begin{proof}
By induction on $n$.  It is straightforward to verify that the result holds for $n=2,3,4$ by iterating through the algorithms. Assume $n>4$, and that $\LIST{j} = \textsc{Greedy}(P_j)$ and $\REVLIST{j} = \textsc{Greedy}(L_j)$ for $2\le j<n$. We begin by showing $\LIST{n} = \textsc{Greedy}(P_n)$, breaking the proof into each of the four stages of a call to \Gen{}$(n-1, 1, 0)$ starting with $P_n$. \\

\noindent \underline{S1:} Since $n>4$ and $s_1=1$, \Gen{}$(n-1, 1, 0)$ is executed. By our inductive hypothesis, $\LIST{n-1} = \textsc{Greedy}(P_{n-1})$. These must be the first trees for \textsc{Greedy}$(P_n)$, as any edge move involving $v_n v_{n-1}$ or $v_n v_\infty$ is larger than any edge move made by \textsc{Greedy}$(P_{n-1})$. Since \textsc{Greedy}$(P_{n-1})$ halts, it must be that no edge move of $T_{n-1}$ is possible. So \textsc{Greedy}$(P_n)$ must make the next smallest edge move, which is $T_n - v_n v_{n-1} + v_n v_\infty$. Since $T_n$ is a spanning tree, it follows that $T_n - v_n v_{n-1} + v_n v_\infty$ is also a spanning tree (and has not been generated yet), and therefore the edge move is valid. At this point, \Gen{}$(n, 1, 0)$ also makes this edge move, by line 13.\\
\begin{sloppypar}
\noindent \underline{S2:} \RevGen{}$(n-1, 1, 0)$ ($T_{n-1} = L_{n-1}$) is then executed. By our inductive hypothesis, $\REVLIST{n} = \textsc{Greedy}(L_{n-1})$. Since \textsc{Greedy}$(L_{n-1})$ halts, it must be that no edge moves of $T_{n-1}$ are possible. At this point, $T_{n-1} = P_{n-1}$ because \RevGen{}$(n-1, 1, 0)$ was executed. The smallest edge move now remaining is $T_n - v_{n-2} v_{n-1} + v_n v_{n-1}$. This results in $T_n = P_{n-2} + v_n v_{n-1} + v_n v_\infty$, which is a spanning tree that has not been generated. So, \textsc{Greedy}$(P_n)$ must make this move. \Gen{}$(n, 1, 0)$ also makes this move, by line 15. So, $\LIST{n}$ must equal \textsc{Greedy}$(P_n)$ up to the end of S2.\\

\noindent \underline{S3:} Next, \Gen{}$(n-2, 1, 1)$ starting with $T_{n-2} = P_{n-2}$ is executed. Since $varEdge = 1$, $v_{n-2} v_{n-1}$ is added instead of $v_{n-2} v_\infty$. \textsc{Greedy}$(P_n)$ also adds $v_{n-2} v_{n-1}$ instead of $v_{n-2} v_\infty$ since $v_{n-2} v_{n-1}$ is smaller than $v_{n-2} v_\infty$ and this edge move results in a tree not yet generated. Other than the difference in this one edge move, which occurs outside the scope of $T_{n-2}$, \Gen{}$(n-2, 1, 0)$ and \Gen{}$(n-2, 1, 1)$ (both starting with $T_{n-2}=P_{n-2}$) make the same edge moves. Since we also know that $\LIST{n-2} = \textsc{Greedy}(P_{n-2})$ by the inductive hypothesis, it follows that $\LIST{n}$ continues to equal \textsc{Greedy}$(P_n)$ after line 16 of \Gen{}$(n,1,0)$ is executed. We know that $T_{n-2} = L_{n-2}$ after \Gen{}$(n-2, 1, 0)$. However, $T_{n-2} = L_{n-2} - v_{n-2} v_\infty + v_{n-2} v_{n-1}$ instead because \Gen{}$(n-2, 1, 1)$ was executed ($varEdge=1$). It must be that no edge moves of $T_{n-2}$ are possible because \textsc{Greedy}$(P_{n-2})$ (and \Gen{}$(n-2, 1, 1)$) halted. The smallest edge move now remaining is $T_n - v_{n-2} v_{n-1} + v_{n-2} v_\infty$. This results in $T_{n-2} = L_{n-2}$. Also, $T_n = T_{n-2} + v_n v_{n-1} + v_n v_\infty$ is a spanning tree since $T_{n-2}$ is a spanning tree of $F_{n-2}$. So \textsc{Greedy}$(P_n)$ makes this move. \Gen{}$(n, 1, 0)$ also makes this move, by line 17, and thus $\LIST{n} = \textsc{Greedy}(P_n)$ up to the end of S3. \\
\end{sloppypar}
\noindent \underline{S4:} Finally, \RevGen{}$(n-2, 0, 0)$  starting with $T_{n-2} = L_{n-2}$ is executed. By our inductive hypothesis, $\REVLIST{n-2} = \textsc{Greedy}(L_{n-2})$. From lines 15-18 of Algorithm~\ref{alg: RevGen}, it is clear that \RevGen{}$(n-2, 0, 0)$ and \RevGen{}$(n-2, 1, 0)$ make the same edge moves until \RevGen{}$(n-2, 0, 0)$ finishes executing. So, by the inductive hypothesis, the listings produced by \RevGen{}$(n-2, 0, 0)$ and \textsc{Greedy}$(L_{n-2})$ are the same until this point, which is where \Gen{}$(n, 1, 0)$ finishes execution. By Lemma~\ref{EdgeMovesLemma} we have that $|\LIST{n}| = \spt{n}$. Therefore, \textsc{Greedy}$(P_n)$ has also produced this many trees, and each tree is unique. Thus, it must be that all $\spt{n}$ trees of $F_n$ have been generated. Thus, \textsc{Greedy}$(P_n)$ also halts.

Since $\LIST{n}$ and \textsc{Greedy}$(P_n)$ start with the same tree, produce the same trees in the same order, and halt at the same place, it follows that $\LIST{n} = \textsc{Greedy}(P_n)$. We now prove that $\REVLIST{n} = \textsc{Greedy}(L_n)$. By using our inductive hypothesis and the same arguments as previously, we see that this results hold within the four stages of \RevGen{}$(n, 1, 0)$ when starting with $L_n$(namely, \Gen{}$(T_{n-2}, 0, 0)$, \RevGen{}$(T_{n-2}, 1, 1)$, \Gen{}$(T_{n-1}, 1, 0)$, and \RevGen{}$(T_{n-1}, 1, 0)$). 
However, we must still prove that $\REVLIST{n}$ matches $\textsc{Greedy}(L_n)$ as we move between stages.\\

\noindent \underline{After S4:} After \Gen{}$(n-2, 0, 0)$, $T_n = L_{n-2} + v_n v_{n-1} + v_n v_\infty$. Since $v_{n-2} v_\infty \in T_n$ and no edge moves of $T_{n-2}$ are possible (because \textsc{Greedy}$(L_{n-2})$ halted), \textsc{Greedy} makes the next smallest edge move which is $T_n - v_{n-2} v_\infty + v_{n-2} v_{n-1}$. \RevGen{}$(n, 1, 0)$ also makes this move here, by line 11.\\

\noindent \underline{After S3:} After \RevGen{}$(n-2, 1, 1)$, $T_n = P_{n-2} + v_n v_{n-1} + v_n v_\infty$. No edge moves of $T_{n-2}$ are possible at this point. Therefore, since $v_n v_{n-1} \in T_n$ and $v_{n-2} v_{n-1} \not \in T_n$, \textsc{Greedy} makes the smallest possible edge move which is $T_n - v_n v_{n-1} + v_{n-2} v_{n-1}$. \RevGen{}$(n, 1, 0)$ also makes this move here, by line 13. \\

\noindent \underline{After S2:} Finally, after \Gen{}$(n-1, 1, 0)$, $T_n = L_{n-1} + v_n v_\infty$.  No edge moves of $T_{n-1}$ are possible at this point. Therefore, \textsc{Greedy} must make the only remaining edge move, which is $T_n - v_n v_\infty + v_n v_{n-1}$. \RevGen{}$(n, 1, 0)$ also makes this move here, by line 17.

Since \textsc{Greedy}$(L_n)$ and \RevGen{}$(n, 1, 0)$ start with the same tree, produce the same trees in the same order, and halt at the same place, it follows that $\REVLIST{n} = \textsc{Greedy}(L_n)$.
\end{proof}
%
%

Because \textsc{Greedy}$(P_n)$ generates unique spanning trees of $F_n$, Lemma~\ref{EdgeMovesLemma} together with Lemma~\ref{Gen=GreedyTheorem} implies the following. 
%
\begin{lemma}
For $n \geq 2$,  $\LIST{n}$ = \textsc{Greedy}$(P_n)$ is a pivot Gray code listing for the spanning trees of $F_n$.
\end{lemma}

It remains to prove how efficiently our pivot Gray codes can be generated. To store the global tree $T$, the algorithms \Gen{} and \RevGen{} can employ an adjacency list model where each edge $uv$ is associated only with the smallest labeled vertex $u$ or $v$.  This means $v_\infty$ will never have any edges associated with it, and every other vertex will have at most 3 edges in its list.  Thus the tree $T$ requires at most $O(n)$ space to store, and edge additions and deletions can be done in constant time.  
The next result completes the proof of Theorem~\ref{thm:main}.

\begin{lemma}
For $n\geq 2$, $\LIST{n}$ and $\REVLIST{n}$ can be generated in $O(1)$-amortized time using $O(n)$ space.
\end{lemma}

\begin{proof}
For each call to \Gen{}$(n, s_1, varEdge)$ where $n>3$, there are at most four recursive function calls, and at least two new spanning trees generated.  Thus, the total number of recursive calls made is at most twice the number of spanning trees generated.  Each edge addition and deletion can be done in constant time as noted earlier. Thus each recursive call requires a constant amount of work, and hence the overall algorithm will run in $O(1)$-amortized time.  There is a constant amount of memory used at each recursive call and the recursive stack goes at most $n-3$ levels deep; this requires $O(n)$ space.  As mentioned earlier, the global variable $T$ stored as adjacency lists also requires $O(n)$ space.
\end{proof}

\subsection{Proof of Lemma~\ref{lem:last}}

\begin{sloppypar}
We prove that $L_n = Last_n$ for $n \geq 2$ by induction on $n$ and tracing the routines {\sc Gen} and {\sc RevGen}.  By definition of $Last_n$, the result holds for $n=2,3,4$. Assume that $L_j = Last_j$ for $2 \leq j < n$.  Recall that $L_n$ is the last tree generated by a call to \Gen{}$(n,1,0)$ when starting with the tree $P_n$.  Tracing this routine  following the proof of  Lemma~\ref{Gen=GreedyTheorem}, 
the current spanning tree when 
calling \RevGen{}$(n-2, 0, 0)$ (on line 18) is $L_{n-2} + v_n v_{n-1} + v_n v_\infty$.  Thus from the definition of $Last_n$, we must show that the last tree of \RevGen{}$(n-2, 0, 0)$ when starting with $L_{n-2}$ is $Last_{n-3} + v_{n-2}v_\infty$.
%

Since \RevGen{}$(k, 1, 0)$ starting with $L_k$ is the reversal of \Gen{}$(k, 1, 0)$ starting with $P_k$, then the last tree of S2 of \RevGen{}$(k, 1, 0)$ must be the first tree of S2 of \Gen{}$(k, 1, 0)$.  The last tree of the recursive call S1 of \Gen{}$(k, 1, 0)$ when starting with $P_k$ is $Last_{k-1} + v_k v_{k-1}$ because, by the inductive hypothesis, $L_{k-1} = Last_{k-1}$, and $v_k v_{k-1} \in P_k$. Then, the edge move  made by line 13 of \Gen{}$(k, 1, 0)$  removes $v_k v_{k-1}$ and adds $v_k v_\infty$. It follows that the first tree of S2 of \Gen{}$(k, 1, 0)$, and equivalently the last tree of S2 of \RevGen{}$(k, 1, 0)$, which is also the last tree of \RevGen{}$(k, 0, 0)$ when starting with $L_k$, is $Last_{k-1} + v_k v_\infty$.  Thus, the last tree of \RevGen{}$(n-2, 0, 0)$ starting with $L_{n-2}$ is $Last_{n-3} + v_{n-2} v_\infty$, as desired.

\end{sloppypar}



\section{Conclusion} \label{sec:summary}

We answer each of the three Research Questions outlined in Section~\ref{sec:intro} in the affirmative for the
fan graph, $F_n$.  First, we discovered a greedy algorithm that exhaustively listed all spanning trees of $F_n$ experimentally for small $n$ with an easy to define starting tree.  
We then studied this listings which led to a recursive construction producing the same listing that runs in $O(1)$-amortized time using $O(n)$ space.  We also proved that the greedy algorithm does in fact exhaustively list all spanning trees of $F_n$ for all $n\geq 2$, by demonstrating the listing is equivalent to the aforementioned recursive algorithm.  It is the first greedy algorithm known to exhaustively list all spanning trees for a non-trivial class of graphs. Finally, we provided $O(n)$-time ranking and unranking algorithms for our listings, assuming the unit cost RAM model. It remains an interesting open problem to answer the research questions for other classes of graphs including the wheel, $n$-cube, and complete graph.

\subsection{Final comment: The wheel}
The wheel $W_n$ is obtained by adding the single edge $v_2v_n$ to $F_n$. With the addition of this single edge, we were unable to find a greedy algorithm to list all the spanning trees of $W_n$ in a pivot Gray code order.  However, we were able to adapt the recursive algorithm for the spanning trees of $F_n$ to obtain a pivot Gray code for $W_n$ by appropriately inserting the spanning trees of $W_n$ that contain the \emph{wheel edge} $v_2v_n$.

\begin{figure}
\begin{center}
  \includegraphics[scale = 0.22]{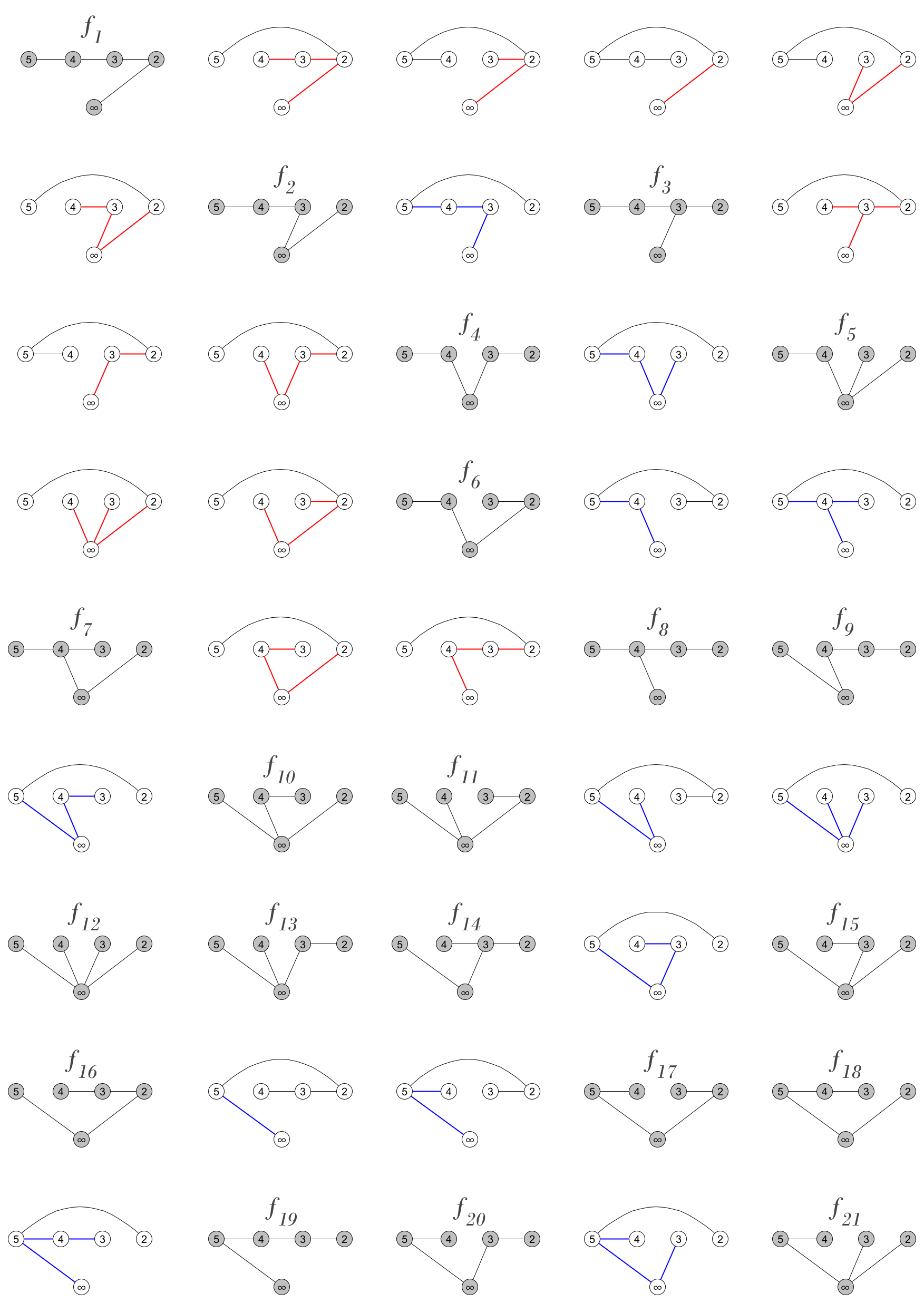}
  \caption{A pivot Gray code listing for the $W_5$ obtained by inserting the spanning trees with the edge $v_2v_5$ in between the trees of $\LIST{5}$ (colored with grey vertices). The label $f_j$ denotes the $j$th tree of $\LIST{5}$. Blue and red are used to highlight the edges of a left or right subgraph, respectively.}
  \label{fig:W5PivotExample}
\end{center}
\end{figure}

Figure~\ref{fig:W5PivotExample} provides an example of a pivot Gray code listing for the spanning trees of $W_5$ obtained by inserting the trees containing the edge $v_5 v_2$ into the listing for $\LIST{5}$.   
Note that all the trees containing the wheel edge $v_5 v_2$ contain subgraphs corresponding to the spanning trees of $F_4$, $F_3$, or $F_2$. For example, the second tree in the first row contains the first spanning tree of $\LIST{4}$ as a subgraph (highlighted in red) on the vertices $v_4, v_3, v_2,$ and $v_\infty$. The third tree of the second row also contains the first tree of $\LIST{4}$, except this time it appears as a subgraph (highlighted in blue) on the vertices $v_5, v_4, v_3,$ and $v_\infty$. We now introduce some terminology to differentiate between these two cases. A tree, $T$, with the wheel edge contains a \textit{right subgraph} if $v_2$ is connected to $v_\infty$ when the edge $v_n v_2$ is removed from $T$. Similarly, $T$ contains a \textit{left subgraph} if $v_n$ is connected to $v_\infty$ when the edge $v_n v_2$ is removed from $T$. As a further example of trees that contain right subgraphs, see the third and fourth tree on the first row, which contain the first tree of $\LIST{3}$ and the first tree of $\LIST{2}$, respectively.


When inserting the trees with the wheel edge, it is  straightforward to fit in the trees that contain right subgraphs due to the recursive nature of \Gen{}. Since the trees of S1 of \Gen{}$(n, 1, 0)$ are all of the trees of \Gen{}$(n-1, 1, 0)$ with the edge $v_n v_{n-1}$ added, we can add trees containing the wheel edge and these right subgraphs as intermediate trees in between the trees of S1 of $\LIST{n}$ by removing the edge $v_n v_{n-1}$ and adding $v_n v_2$. We can further insert the trees containing smaller subgraphs (like the third and fourth tree on the first row) by rotating edges along the path, as seen in between $f_1$ and $f_2$ of Figure~\ref{fig:W5PivotExample}. Note that in between the fourth and fifth tree of the first row, we make the original edge move between $f_1$ and $f_2$, and then rotate edges back along the path to obtain $f_2$. 

Unfortunately, the trees containing left subgraphs do not follow a nice recursive pattern. However, we are still able to insert them appropriately by considering two cases. First, when $v_2$ is the pivot vertex, we can insert a single tree containing a left subgraph as an intermediate edge move. As an example of this, see the tree in between $f_2$ and $f_3$ of Figure~\ref{fig:W5PivotExample}. The other case is when $v_2 v_\infty$ is present and $v_2$ is not the pivot vertex, in which case we can insert two trees. For example, in between the trees labeled $f_6$ and $f_7$ in Figure~\ref{fig:W5PivotExample}. Note that we require $v_2 v_\infty$ to be present so that we can remove it and add $v_n v_2$ as an intermediate move. If we instead replaced either $v_2 v_\infty$ or $v_2 v_3$ with $v_n v_2$, then we could end up with duplicate trees. Also note that generating right subgraphs takes precedence over generating left subgraphs. For example, even though our first case is satisfied in between $f_7$ and $f_8$, there are still trees with right subgraphs that can be generated, so we do not generate a tree with a left subgraph. Finally, note that inserting trees in the ways we have described does not change the relative order that the trees of $\LIST{n}$ appear.  


\footnotesize
\bibliographystyle{abbrv}
\bibliography{References}
\normalsize

\newpage
\appendix
\section{C Code}
\label{app:code}

\scriptsize
\begin{code}
#include <stdio.h>
#include <stdlib.h>
#define MAX_N 30

int n;
int tree[MAX_N+2][MAX_N+2];     // Adjacency matrix of a spanning tree
long long int fib[2*MAX_N+1];   // Stores the Fibonacci numbers
long long int numTrees = 1;     // Number of trees generated

void ReverseGen(int k, int S1, int varEdge);

//-------------------------------------------------
void PrintMove(int v, int old, int new) {
    printf("Move #%lld: -(%d, %d) +(%d, %d)\n", numTrees, v-1, old-1, v-1, new-1);
}
//-------------------------------------------------
void PrintTree() { // Prints edge list of tree

  for (int i = 2; i < n+1; i++) {
    for (int j = i+1; j < n+2; j++) {
      if (tree[i][j] == 1) {
        printf("%d %d\n", i-1, j-1);
      }
    }
  } printf("\n");

}
//-------------------------------------------------
int tF(int k) { return fib[2*k - 2]; } // Calculates t(F_k)

//-------------------------------------------------
void CreateStartTree() { // Creates adjacency matrix of P_n
    tree[2][n+1] = tree[n+1][2] = 1;
    for (int i = 3; i < n+1; i++) {
        tree[i][i-1] = tree[i-1][i] = 1;
    }
}
//-------------------------------------------------
void CreateLastTree(int k) { // Creates adjacency matrix of L_n
    if (k == 2) {
        tree[2][n+1] = tree[n+1][2] = 1;
    } else if (k == 3) {
        tree[2][3] = tree[3][2] = 1;
        tree[3][n+1] = tree[n+1][3] = 1;
    } else if (k == 4) {
        tree[2][3] = tree[3][2] = 1;
        tree[3][4] = tree[4][3] = 1;
        tree[4][n+1] = tree[n+1][4] = 1;
    } else if (k > 4){
        tree[k][k-1] = tree[k-1][k] = 1;
        tree[k][n+1] = tree[n+1][k] = 1;
        tree[k-2][n+1] = tree[n+1][k-2] = 1;
        CreateLastTree(k-3);
    }
}
//-------------------------------------------------
void CreateFib() { // Populates the Fibonacci array
    fib[1] = fib[2] = 1;
    for (int i = 3; i <= 2*(MAX_N-1); i++) fib[i] = fib[i-1] + fib[i-2];
}
//-------------------------------------------------
int Rank(int k) {
    
    // Base cases
    if (k == 3) { // F_3
        if (tree[k][k-1] == 1 && tree[k][n+1] == 1) return 3;
        else if (tree[k][n+1] == 1) return 2;
        else if (tree[k][k-1] == 1) return 1;
    } 
    else if (k == 2) { // F_2
        if (tree[k][n+1] == 1) return 1;
    }
    
    if (tree[k][k-1] == 1 && tree[k][n+1] == 1) { // Both edges present
        
        if (tree[k-2][n+1] == 1) return 2*tF(k-1) + 2*tF(k-2) - Rank(k-2) + 1*(k!=4); //  S4
        else if (tree[k-2][k-1] == 1) {
            
            // Delete e_4, Add e_3 and continue as normal
            tree[k-2][k-1] = tree[k-1][k-2] = 0;
            tree[k-2][n+1] = tree[n+1][k-2] = 1;
        }      
        return 2*tF(k-1) + Rank(k-2) + 1*(k==4); // S3
    } 
    else if (tree[k][k-1] == 1) return Rank(k-1); // S1
    else if (tree[k][n+1] == 1) return 2*tF(k-1) - Rank(k-1) + 1; // S2
    return 0; 
}
//-------------------------------------------------
void Unrank(int k, int rank, int replaceEdge) {
    
    // Base cases
    if (k == 2) { // F_2
        if (replaceEdge == 1) tree[3][2] = tree[2][3] = 1;
        else tree[n+1][2] = tree[2][n+1] = 1;
        return;
    }
    else if (k == 3) { // F_3
        if (rank == 1) {
            tree[n+1][2] = tree[2][n+1] = 1;   tree[2][3] = tree[3][2] = 1;
        }  
        else if (rank == 2) {
            tree[2][n+1] = tree[n+1][2] = 1;
            if (replaceEdge == 1) tree[3][4] = tree[4][3] = 1;
            else tree[n+1][3] = tree[3][n+1] = 1;
        } 
        else if (rank == 3) {
            tree[3][2] = tree[2][3] = 1;
            if (replaceEdge == 1) tree[3][4] = tree[4][3] = 1;
            else tree[3][n+1] = tree[n+1][3] = 1;
        }
        return;
    }   
    if (rank <= tF(k-1)) { // S1 - Add e_1  
        tree[k][k-1] = tree[k-1][k] = 1;       
        Unrank(k-1, rank, 0);
    }
    else if (rank <= 2*tF(k-1)) { // S2 - Add e_2
        if (replaceEdge == 1) tree[k][k+1] = tree[k+1][k] = 1;
        else tree[k][n+1] = tree[n+1][k] = 1;        
        Unrank(k-1, 2*tF(k-1) - rank + 1, 0);
    }
    else if (rank <= 2*tF(k-1) + tF(k-2)) { // S3 - Add both edges 
        tree[k][k-1] = tree[k-1][k] = 1;
        if (replaceEdge == 1) tree[k][k+1] = tree[k+1][k] = 1;
        else tree[k][n+1] = tree[n+1][k] = 1;      
        Unrank(k-2, rank - 2*tF(k-1), 1*(k!=4));
    }
    else if (rank <= 3*tF(k-1) - tF(k-2)) { // S4 - Add both edges        
        tree[k][k-1] = tree[k-1][k] = 1;
        if (replaceEdge == 1) tree[k][k+1] = tree[k+1][k] = 1;
        else tree[k][n+1] = tree[n+1][k] = 1;      
        Unrank(k-2, 2*tF(k-1) + 2*tF(k-2) - rank + 1, 1*(k==4));
    }
}
//-------------------------------------------------
void Replace(int v, int old, int new) { // Delete (v, old), Add (v, new)
    PrintMove(v, old, new);
    tree[v][old] = tree[old][v] = 0;
    tree[v][new] = tree[new][v] = 1; numTrees++;
}
//-------------------------------------------------
void Gen(int k, int S1, int varEdge) {
    
    if (k == 2) { // F_2 base case
        if (varEdge == 1) Replace(2, n+1, 3);
    } 
    else if (k == 3) { // F_3 base case
        if (S1 == 1) {
            if (varEdge == 0) Replace(3, 2, n+1);
            else Replace(3, 2, 4); // S3
        }
        Replace(2, n+1, 3);
    } 
    else {
        if (S1 == 1) {
            Gen(k-1, 1, 0);
            if (varEdge == 0) Replace(k, k-1, n+1);
            else Replace(k, k-1, k+1); // S3
        }
        ReverseGen(k-1, 1, 0);
        Replace(k-1, k-2, k);
        Gen(k-2, 1, 1);
        if (k > 4) Replace(k-2, k-1, n+1);
        ReverseGen(k-2, 0, 0);
    } 
}
//-------------------------------------------------
void ReverseGen(int k, int S1, int varEdge) {
    
    if (k == 2) {
        if (varEdge == 1) Replace(2, 3, n+1);
    } 
    else if (k == 3) {
        Replace(2, 3, n+1);
        if (S1 == 1) {
            if (varEdge == 0) Replace(3, n+1, 2);
            else Replace(3, 4, 2);
        }
    } 
    else {
        Gen(k-2, 0, 0);
        if (k > 4) Replace(k-2, n+1, k-1);
        ReverseGen(k-2, 1, 1);
        Replace(k-1, k, k-2);
        Gen(k-1, 1, 0);
        if (S1 == 1) {
            if (varEdge == 0) Replace(k, n+1, k-1);
            else Replace(k, k+1, k-1);
            ReverseGen(k-1, 1, 0);
        }
    } 
}
//-------------------------------------------------
int main() {  
  int choice, rank, v1, v2;

  // User input and error checking
  printf(" ###################################################################################");
  printf("####################################################################################\n\n");
  printf(" This program provides functionality to list the spanning trees of the Fan graph");
  printf(" in a pivot Gray code order, rank a tree in the listing, or unrank a tree in the listing.\n");
  printf(" The vertices on the path are labeled 1 to n-1, and the universal vertex is labeled n.\n\n");
  printf(" ###################################################################################");
  printf("####################################################################################\n\n");

  printf("  1. Pivot Gray code generation (GEN)\n");
  printf("  2. Pivot Gray code generation in reverse order of option 1 (REVGEN)\n");
  printf("  3. Rank a tree in the listing generated by option 1\n");
  printf("  4. Unrank a tree in the listing generated by option 1\n");
  printf("  Enter selection: ");  scanf("%d", &choice);

  if (choice < 1 || choice > 4) {
    printf("Error: Invalid choice.\n");
    exit(0);
  }
  printf("Input n: "); scanf("%d", &n);
  if (n > MAX_N) {
    printf("Error: n is too big. Please try n <= 30.\n");
  }
  CreateFib();

  if (choice == 1) { // GEN

    CreateStartTree();
    printf("\n##### GEN ####\n");
    Gen(n, 1, 0);
    printf("Number of spanning trees of F_%d: %lld\n", n, numTrees);

  } else if (choice == 2) {

    CreateLastTree(n);
    printf("\n#### REVGEN ####\n");
    ReverseGen(n, 1, 0);
    printf("Number of spanning trees of F_%d: %lld\n", n, numTrees);

  } else if (choice == 3) { // RANK

    printf("Enter the edges of the spanning tree in format 'v1 v2'. ");
    printf("If you input edge (v1, v2), do not input edge (v2, v1). ");
    printf("Use labels 1 to n-1 for the vertices on the path (from right to left)");
    printf(", and label n for the universal vertex. ");
    printf("Warning: no error checking is done.\n");
    for (int i = 1; i <= n-1; i++) {
      printf("Edge %d: ", i);
      scanf("%d %d", &v1, &v2);
      tree[v1+1][v2+1] = tree[v2+1][v1+1] = 1;
    }
    printf("Rank of inputted tree in listing for GEN is #%d.\n", Rank(n));

  } else if (choice == 4){ // UNRANK

    printf("Enter rank (between 1 and %lld): ", fib[2*(n-1)]);
    scanf("%d", &rank);
    if (rank < 1 || rank > tF(n)) {
      printf("Error: Invalid input.\n"); exit(0);
    }
    Unrank(n, rank, 0);
    printf("\nTree #%d of GEN for F_%d is: \n\n", rank, n);
    PrintTree();
  }

  return 0;
}
\end{code}

\end{document}